\documentclass[conference]{IEEEtran}
\usepackage[utf8]{inputenc}
\usepackage{amsmath}
\usepackage{amsfonts}
\usepackage{amsthm}
\usepackage{tikz}

\newtheorem{thm}{Theorem}

\newtheorem{lem}{Lemma}

\newtheorem{con}{Conjecture}

\title{The Courtade-Kumar Most Informative Boolean Function Conjecture and a Symmetrized Li-M\'edard Conjecture are Equivalent}
\author{Leighton Pate Barnes and Ayfer \"Ozg\"ur\\
Stanford University, Stanford, CA 94305\\
 Email: \{lpb, aozgur\}@stanford.edu}

\begin{document}

\maketitle

\begin{abstract}
We consider the Courtade-Kumar most informative Boolean function conjecture for balanced functions, as well as a conjecture by Li and M\'edard that dictatorship functions also maximize the $L^\alpha$ norm of $T_pf$ for $1\leq\alpha\leq2$ where $T_p$ is the noise operator and $f$ is a balanced Boolean function. By using a result due to Laguerre from the 1880's, we are able to bound how many times an $L^\alpha$-norm related quantity can cross zero as a function of $\alpha$, and show that these two conjectures are essentially equivalent.
\end{abstract}

\section{Introduction}
In 2014, Courtade and Kumar published \cite{courtade_kumar} in which they introduced a conjecture about which Boolean function maximizes mutual information when applied to a noisy input. More concretely, they considered the following scenario. Suppose $X$ is uniformly distributed on the $n$-dimensional Hamming cube $\{0,1\}^n$. Let $Y$ be a noisy copy of $X$ which is the output of a memoryless binary symmetric channel with crossover probability $0<p<\frac{1}{2}$ when $X$ is the input. The optimization problem of interest is to find which function $f:\{0,1\}^n \to \{0,1\}$ maximizes the mutual information $I(f(Y);X)$ (or equivalently $I(f(X);Y)$).

Courtade and Kumar conjectured that
\begin{equation} \label{ck_conj} I(f(Y);X) \leq 1 - h(p) \end{equation}
where $h(p) = -p\log p - (1-p)\log (1-p)$ is the binary entropy function. This mutual information is achieved when
$$f(Y) = f(Y_1,\ldots,Y_n) = Y_i$$
for any $i=1,\ldots n$. Such functions are called \emph{dictatorship} functions. In other words, Courtade and Kumar conjecture that there is no more informative Boolean function than simply taking one of the coordinate values.

In recent years there has been substantial interest in solving this problem, as well as some partial results. Chandar and Tchamkerten show the bound
\begin{equation} I(f(Y);X) \leq (1-2p)^2 \label{eq:bound1} \end{equation}
in \cite{chandar}. Ordentlich, Shayevitz and Weinstein prove the bound
$$I(f(Y);X) \leq \frac{\log(e)}{2}(1-2p)^2+9\left(1-\frac{\log(e)}{2}\right)(1-2p)^4$$
for $\frac{1}{2}\left(1-\frac{1}{\sqrt{3}}\right) \leq p \leq \frac{1}{2}$ in \cite{or}, which is tighter than \eqref{eq:bound1} for $\frac{1}{3}\leq p \leq \frac{1}{2}$. Using this improved bound, as well as a strengthened version of ``Mrs. Gerber's Lemma'', Samorodnitsky was able to prove the conjecture is true for all $p \geq \frac{1}{2} - \delta$ where $\delta > 0$ is some absolute constant in \cite{alex}.

Along with these partial results, there have also been some related conjectures made in \cite{muriel,hel_conj}. In these works it is conjectured that dictatorship functions maximize other functionals of $f$, in such a way that the conjectures are stronger than the Courtade-Kumar conjecture -- meaning that if the conjectures hold then \eqref{ck_conj} must also hold. Of particular interest to us is the conjecture from Li and M\'edard in \cite{muriel} that focuses on \emph{balanced} Boolean functions $f$ that have fixed mean
$$\frac{1}{2^n}\sum_{y\in\{0,1\}^n}f(y) = \frac{1}{2} \; .$$
They conjecture that among all such balanced functions $f$, the $L^\alpha$-norm of the function
\begin{equation} \label{eq:noise_op} T_pf(x) = \mathbb{P}(f(Y)=1|X=x) \end{equation}
is maximized by dictatorship functions when $1\leq\alpha\leq2$. Li and M\'edard show that if their conjecture is true, then \eqref{ck_conj} must also be true for all balanced functions $f$. Our main contribution in this paper is to show that a slightly modified converse statement also holds. In particular, if we ``symmetrize'' the Li-M\'edard conjecture by including both $T_pf(x)$ and $1-T_pf(x)$ in the calculation of the $L^\alpha$-norm, then the Courtade-Kumar conjecture for balanced functions $f$ also implies this symmetrized conjecture. In this sense, the two conjectures are equivalent.

In order to show this equivalence, we study the quantity
\begin{equation}
\label{eq:alphanorm} N_\alpha(f) = \sum_{x\in\{0,1\}^n} (T_pf(x))^\alpha
\end{equation}
or its symmetrized version
\begin{equation}
\label{eq:alphanorm_sym} N_\alpha^\mathsf{sym}(f) = \sum_{x\in\{0,1\}^n} (T_pf(x))^\alpha + (1-T_pf(x))^\alpha
\end{equation}
for any $\alpha\in\mathbb{R}$. Note that these quantities are not technically norms of $T_pf$ except in the case of \eqref{eq:alphanorm} when $\alpha=1$. However, \eqref{eq:alphanorm} can be thought of as the $\alpha$-power of the $L^\alpha$-norm when $\alpha\geq1$. Letting $f_0$ be a dictatorship function, we define
\begin{equation} \label{eq:g} g_f(\alpha) = N_\alpha(f) - N_\alpha(f_0) \end{equation}
and the corresponding
\begin{equation} \label{eq:g_sym} g_f^\mathsf{sym}(\alpha) = N_\alpha^\mathsf{sym}(f) - N_\alpha^\mathsf{sym}(f_0) \; .\end{equation}
By using a result due to Laguerre from the 1880's \cite{laguerre}, we will show that both $g_f(\alpha)$ and $g_f^\mathsf{sym}(\alpha)$ can have at most four zeros (unless, of course, they are identically zero). This allows us to relate the validity of the Courtade-Kumar conjecture, which can be thought of as a local property of $g_f^\mathsf{sym}(\alpha)$ around $\alpha=1$, to the global properties of $g_f^\mathsf{sym}(\alpha)$ such as where the zeros are located. We believe this is a powerful insight that may shed further light on the validity of these sorts of conjectures.

\section{Main Results}
In this section we state precisely the relevant conjectures and how they are equivalent. In order to make sense of the first conjecture below, we first need to analyze the mutual information quantity of interest $I(f(Y);X)$. For balanced $f$ we can write this quantity as follows:
\begin{align} \label{eq:rewrite}
I(f(Y);X) & = H(f(Y)) - H(f(Y)|X) \nonumber \\
& = 1 - \frac{1}{2^n}\sum_{x\in\{0,1\}^n}h\left(\mathbb{P}(f(Y)=1|X=x)\right) \nonumber \\
& = 1 + \frac{1}{2^n}\sum_{x\in\{0,1\}^n} \bigg( T_pf(x)\log T_pf(x) \nonumber \\
& \quad \quad + (1-T_pf(x))\log(1-T_pf(x)) \bigg) \; .
\end{align}
Looking at \eqref{eq:rewrite}, it makes sense to talk about an ``unsymmetrized'' Courtade-Kumar conjecture where we consider which balanced Boolean function maximizes only the part $\sum_{x\in\{0,1\}^n}  T_pf(x)\log T_pf(x) \; .$ Recall that we say a function $f:\{0,1\}^n \to \{0,1\}$ is balanced if $\mathbb{E}[f(Y)] = \frac{1}{2}$ and that $f_0$ represents a dictatorship function.
\begin{con}[unsymmetrized Courtade-Kumar]\label{con1}
For any balanced $f$,
$$\sum_{x\in\{0,1\}^n}  T_pf(x)\log T_pf(x) \leq \sum_{x\in\{0,1\}^n}  T_pf_0(x)\log T_pf_0(x) \; .$$
\end{con}
\begin{con}[Courtade-Kumar]\label{con2}
For any balanced $f$,
$$I(f(Y);X) \leq 1-h(p) \; .$$
\end{con}
\begin{con}[Li-M\'edard]\label{con3}
For any balanced $f$,
$$N_\alpha(f)\leq N_\alpha(f_0)$$
for $1\leq\alpha \leq 2$.
\end{con}
\begin{con}[symmetrized Li-M\'edard]\label{con4}
For any balanced $f$,
$$N^\mathsf{sym}_\alpha(f)\leq N^\mathsf{sym}_\alpha(f_0)$$
for $1\leq\alpha \leq 2$.
\end{con}
We establish equivalences between these conjectures in our main results below.
\begin{thm} \label{thm1}
Conjecture \ref{con1} is true if and only if Conjecture \ref{con3} is true.
\end{thm}
\begin{thm} \label{thm2}
Conjecture \ref{con2} is true if and only if Conjecture \ref{con4} is true.
\end{thm}
\noindent Another relationship between these conjectures is that Conjecture \ref{con1} would imply Conjecture \ref{con2}, and similarly Conjecture \ref{con3} would imply Conjecture \ref{con4}. This is because if $f$ is balanced then $1-f$ is also balanced, and $T_p(1-f) = 1-T_pf$. Furthermore, $1-f_0$ is equivalent to a dictatorship function and will achieve the same value for any of the relevant functionals. We therefore have the following relationships between the various conjectures.

\vspace{.1in}
\begin{center}
\begin{tabular}{c c c}
Conjecture \ref{con1} & $\iff$ & Conjecture \ref{con3} \\
$\Downarrow$ & & $\Downarrow$ \\
Conjecture \ref{con2} & $\iff$ & Conjecture \ref{con4} 
\end{tabular}
\end{center}
\vspace{.1in}

It is worth pointing out that Conjecture \ref{con2} was proven in \cite{alex} in the high-noise case, i.e. for all $\frac{1}{2}-\delta<p<\frac{1}{2}$ where $\delta>0$ is an absolute constant. Therefore our result immediately implies that Conjecture \ref{con4} is also true in this high-noise case. It is also known that Conjecture \ref{con2} holds in the low-noise case (see \cite{courtade_kumar,muriel}) when $0<p<\delta_n$ and $\delta_n$ depends on $n$. In the same way Conjecture \ref{con4} must also hold in this dimensionally-dependent low-noise case.

\section{Proof of Theorems \ref{thm1} and \ref{thm2}}
In this section we develop the machinery needed to prove Theorems \ref{thm1} and \ref{thm2}, as well as finish their proofs. In both cases it is straightforward to show that the Li-M\'edard conjectures imply their corresponding Courtade-Kumar conjectures, but the other direction is more involved. We start with the easy direction.

\subsection{Conjecture \ref{con3} $\implies$ Conjecture \ref{con1}}\label{subsecA}
Taking the derivative of \eqref{eq:alphanorm} with respect to $\alpha$ and evaluating at $\alpha=1$ gives
\begin{align} \label{eq:deriv}
\frac{\partial}{\partial\alpha}N_\alpha(f)\bigg|_{\alpha=1} & = \sum_{x\in\{0,1\}^n} T_pf(x)\log T_pf(x) \; .
\end{align}
Note that \eqref{eq:deriv} matches the quantity from the inequality in Conjecture \ref{con1}. Writing out the derivative in \eqref{eq:deriv} as a difference quotient,
$$\frac{\partial}{\partial\alpha}N_\alpha(f)\bigg|_{\alpha=1} = \lim_{\epsilon \to 0} \frac{N_{1+\epsilon}(f) - N_{1}(f)}{\epsilon} \; .$$
If Conjecture \ref{con3} is true, then $N_{1+\epsilon}(f) \leq N_{1+\epsilon}(f_0)$ for any $0\leq\epsilon\leq1$. Furthermore, by the law of total probability,
\begin{align*}
N_1(f) = & \sum_{x\in\{0,1\}^n}\mathbb{P}(f(Y)=1|X=x) \\
& = 2^n \mathbb{P}(f(Y)=1) \\
& = 2^{n-1} \\
& = N_1(f_0) \; .
\end{align*}
Therefore
$$\frac{\partial}{\partial\alpha}N_\alpha(f)\bigg|_{\alpha=1} \leq \frac{\partial}{\partial\alpha}N_\alpha(f_0)\bigg|_{\alpha=1}$$
and \eqref{eq:deriv} implies Conjecture \ref{con1} must be true.
\subsection{Conjecture \ref{con4} $\implies$ Conjecture \ref{con2}}
This implication is very similar to that of Section \ref{subsecA} above. Differentiating \eqref{eq:alphanorm_sym} with respect to $\alpha$ and evaluating at $\alpha=1$ gives
\begin{align}\label{eq:deriv_sym}
\frac{\partial}{\partial\alpha}N_\alpha^\mathsf{sym}(f)\bigg|_{\alpha=1} & = \sum_{x\in\{0,1\}^n} T_pf(x)\log T_pf(x) \nonumber \\
& \quad \quad + (1-T_pf(x))\log (1-T_pf(x)) \nonumber \\
& = - \sum_{x\in\{0,1\}^n} h(\mathbb{P}(f(Y)=1|X=x)) \; .
\end{align}
Conjecture \ref{con4} implies that $N_{1+\epsilon}^\mathsf{sym}(f) \leq N_{1+\epsilon}^\mathsf{sym}(f_0)$ for $0\leq\epsilon\leq1$ and $N_1^\mathsf{sym}(f) = 2^n = N_1^\mathsf{sym}(f_0)$. Thus
$$\frac{\partial}{\partial\alpha}N^\mathsf{sym}_\alpha(f)\bigg|_{\alpha=1} \leq \frac{\partial}{\partial\alpha}N^\mathsf{sym}_\alpha(f_0)\bigg|_{\alpha=1}$$
and Conjecture \ref{con2} holds.

\subsection{Conjecture \ref{con2} $\implies$ Conjecture \ref{con4}} \label{sec3A}

We approach this implication by contrapositive, and show that if there exists an $1\leq \alpha \leq 2$ such that $N_\alpha^\mathsf{sym}(f) > N_\alpha^\mathsf{sym}(f_0)$, then $I(f(Y);X)> 1-h(p)$. We will need the following two lemmas.

We say that a function $g:\mathbb{R} \to \mathbb{R}$ has a zero at $\alpha$ if $g(\alpha) = 0$, and that it has an $m$th order zero (or a zero with multiplicity $m$) at $\alpha$ if both $g(\alpha)=0$ and the derivatives $g^{(j)}(\alpha) = 0$ for all $j=1,\ldots,m-1$.
\begin{lem}[Laguerre 1883 \cite{laguerre}] \label{lem2}
Suppose $$g(\alpha) = \sum_{i=1}^N A_ie^{c_i \alpha}$$ with $$c_1 < c_2 < \ldots < c_N$$ and $A_i\neq 0$. Then $g$ has at most as many zeros (counting multiplicities) as the number of sign changes in the sequence $A_1,A_2,\ldots,A_N$. 
\end{lem}
\begin{proof} The proof is by induction on the number of sign changes. For the base case note that if each value in the sequence $A_1,\ldots,A_N$ has the same sign, then either $g(\alpha)>0$ or $g(\alpha)<0$ for all $\alpha\in\mathbb{R}$.

For the inductive step, suppose that the sequence $A_1,\ldots,A_N$ changes sign $k$ times as we move from $A_1$ to $A_N$. Suppose that one of the sign changes occurs between indices $l$ and $l+1$, i.e., that $\mathsf{sign}(A_l)\neq\mathsf{sign}(A_{l+1})$. Letting $c_l < b < c_{l+1}$, we can rewrite $g(\alpha)$ as
\begin{align} \label{eq:factor}
g(\alpha) = e^{b\alpha}\sum_{i=1}^N A_ie^{(c_i-b)\alpha} \; .
\end{align}
We isolate the second factor from \eqref{eq:factor} above and define
\begin{align} \label{eq:factor2}
h(\alpha) =\sum_{i=1}^N A_ie^{(c_i-b)\alpha} \; .
\end{align}
Since $e^{b\alpha}>0$, the general Leibniz rule for derivatives of products implies that both $g$ and $h$ will have zeros in the same locations and with the same multiplicities. We will therefore focus on counting the zeros of $h(\alpha)$. Suppose that $h$ has zeros at $\alpha_1,\ldots,\alpha_r$ with corresponding multiplicities $m_1,\ldots,m_r$ and let $M = \sum_{j=1}^r m_r$.

The derivative $h'(\alpha)$ will have $\sum_{j=1}^r \max\{(m_j-1),0\}$ zeros at the points $\alpha_1,\ldots,\alpha_r$, and by Rolle's theorem \cite{mattuck} at least $r-1$ zeros between these points. So in total $h'(\alpha)$ will have at least $M-1$ zeros. Differentiating \eqref{eq:factor2} with respect to $\alpha$,
\begin{align} \label{eq:factor_deriv}
\frac{\partial}{\partial\alpha}\sum_{i=1}^N A_ie^{(c_i-b)\alpha} = \sum_{i=1}^N (c_i-b)A_ie^{(c_i-b)\alpha} \; .
\end{align}
Display \eqref{eq:factor_deriv} fits exactly the form of the function required by the Lemma -- it is a sum of ordered exponentials with $k-1$ sign changes in the coefficients. Hence, by the inductive assumption, the function $h'(\alpha)$ can have at most $k-1$ zeros. Putting this all together we have $M-1 \leq k-1$ and $M\leq k$ as desired.
\end{proof}

\begin{lem} \label{lem1}
For any balanced $f$, $N_2(f) \leq N_2(f_0)$ and $N^\mathsf{sym}_2(f) \leq N^\mathsf{sym}_2(f_0)$ with equality only if $f$ is a dictatorship function.
\end{lem}
\noindent Lemma \ref{lem1} follows immediately from taking the Fourier transform of $T_pf$ and using the Parseval/Rayleigh/Plancherel Theorem. See also \cite{chandar,muriel}. A proof is included in the appendix for completeness

In order to apply Lemma \ref{lem2}, we will consider the function $g^\mathsf{sym}_f(\alpha)$ from \eqref{eq:g_sym}. This function takes the required form since we can write it as
\begin{align}
g^\mathsf{sym}_f(\alpha) & = \sum_{x\in\{0,1\}^n} e^{\alpha\log T_pf(x)} + e^{\alpha\log (1-T_pf(x))} \nonumber \\
& \quad\quad\quad\quad - e^{\alpha\log T_pf_0(x)} - e^{\alpha\log (1-T_pf_0(x))} \; .
\end{align}

\noindent It is clear that $T_pf_0(x) = p$ for $2^{n-1}$ different values of $x$, and $T_pf_0(x) = 1-p$ for the remaining $2^{n-1}$ values of $x$. We will assume without loss of generality that $T_pf(x_0)<p$ or $T_pf(x_0)>1-p$ for some $x_0$ because if this is not the case we would necessarily have
$$(T_pf(x))^\alpha + (1-T_pf(x))^\alpha \leq (T_pf_0(x))^\alpha + (1-T_pf_0(x))^\alpha$$
for all $x$ and $\alpha\geq1$ and there would be nothing to prove. Furthermore, we must have at least one $x$ such that $p<T_pf(x)<1-p$ otherwise this would contradict Lemma \ref{lem1}. The ordered sequence of the $c_i$ (which take the form $\log(T_pf(x))$, $\log(1 -T_pf(x))$, $\log(T_pf_0(x))$, and $\log(1-T_pf_0(x))$) will therefore give rise to a sequence of signs that looks like $+ \; - \; + \; - \; +$  in the coefficients $A_i$. There are therefore exactly four sign changes in the sequence $A_i$, and applying Lemma \ref{lem2}, $g^\mathsf{sym}_f(\alpha)$ can have at most four zeros where $N^\mathsf{sym}_\alpha(f_0) = N^\mathsf{sym}_\alpha(f)$.

Two of the zeros of $g^\mathsf{sym}_f(\alpha)$ must occur at $\alpha = 0$ and $\alpha = 1$. At $\alpha=2$, Lemma \ref{lem1} implies that $g^\mathsf{sym}_f(2) < 0$. However, as $\alpha$ gets large, $N^\mathsf{sym}_\alpha(f_0) < N^\mathsf{sym}_\alpha(f)$ and $g^\mathsf{sym}_f(\alpha)>0$ since $T_p(x_0)>1-p$ or $1-T_p(x_0) > 1-p$. So by the continuity of $N^\mathsf{sym}_\alpha(f)$ with respect to $\alpha$ we must have another zero of $g^\mathsf{sym}_f(\alpha)$ with $\alpha>2$. In a similar way, $N^\mathsf{sym}_\alpha(f_0) < N^\mathsf{sym}_\alpha(f)$ as $\alpha$ approaches $-\infty$.

With all of these restrictions on $g^\mathsf{sym}_f(\alpha)$ in mind, we return to showing that Conjecture \ref{con2} implies Conjecture \ref{con4} by contrapositive. If there exists an $1\leq \alpha \leq 2$ such that $N_\alpha^\mathsf{sym}(f) > N_\alpha^\mathsf{sym}(f_0)$, then the fourth and final zero of $g^\mathsf{sym}_f(\alpha)$ must occur for some $1<\alpha<2$ (since $g^\mathsf{sym}_f(2) < 0$ and the function is continuous). In this case, the function $g^\mathsf{sym}_f(\alpha)$ must take the form depicted in Figure \ref{fig1}, and in particular it must be positive for $1<\alpha<1+\epsilon$ for some $\epsilon>0$. By looking at the difference quotient,
$$\frac{\partial}{\partial\alpha}N^\mathsf{sym}_\alpha(f)\bigg|_{\alpha=1} \geq \frac{\partial}{\partial\alpha}N^\mathsf{sym}_\alpha(f_0)\bigg|_{\alpha=1} \; .$$
We can rule out the possibility that
$$\frac{\partial}{\partial\alpha}N^\mathsf{sym}_\alpha(f)\bigg|_{\alpha=1} = \frac{\partial}{\partial\alpha}N^\mathsf{sym}_\alpha(f_0)\bigg|_{\alpha=1}$$
because it would imply $g^\mathsf{sym}_f(\alpha)$ had a second-order zero at $\alpha=1$ which would contradict Lemma \ref{lem2}. Thus
$$\frac{\partial}{\partial\alpha}N^\mathsf{sym}_\alpha(f)\bigg|_{\alpha=1} > \frac{\partial}{\partial\alpha}N^\mathsf{sym}_\alpha(f_0)\bigg|_{\alpha=1}$$
and because of \eqref{eq:deriv_sym},
$$I(f(Y);X) > 1 - h(p) \; .$$

\begin{figure}
\centerline{\includegraphics[width=8.5cm]{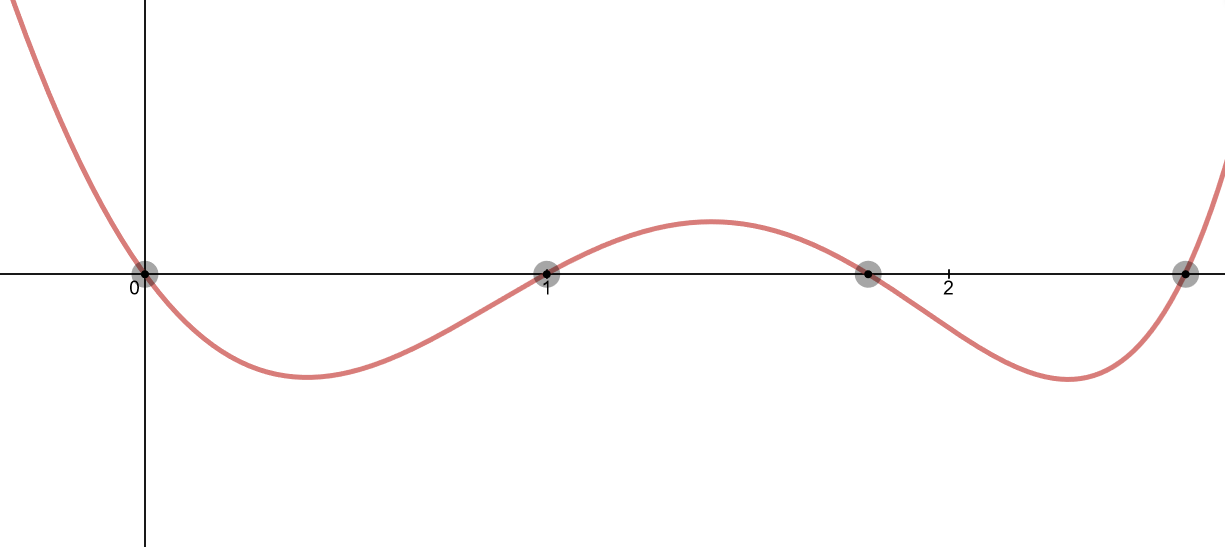}}
\caption{The $g^\mathsf{sym}_f(\alpha)$ curve as a function of $\alpha$ from the proof in Section \ref{sec3A}. The position of the four unique zeros relative to each other are marked, and it is clear where the curve must be  positive and negative. Note that this is not an actual $g^\mathsf{sym}_f(\alpha)$ curve corresponding to a specific Boolean function $f$ -- if that were true it would be a counterexample to the Courtade-Kumar conjecture, none of which are known.}
\label{fig1}
\end{figure}

\subsection{Conjecture \ref{con1} $\implies$ Conjecture \ref{con3}}
The proof of this implication is roughly the same as that of Section \ref{sec3A}. The main difference is that there may exist an $x_0$ with $T_pf(x_0)<p$ or $T_pf(x_0)>1-p$, and since we are not including the corresponding $1-T_pf(x_0)$ term, there might not be any $x_1$ with $T_pf(x_1)>1-p$ or $T_pf(x_1)<p$, respectively. This opens up the possibility that there could be only three sign changes in the coefficients of $g_f(\alpha)$. If there are four sign changes, then the result follows just as in Section \ref{sec3A}. Let us therefore assume that there are three sign changes.

There are two trivial zeros of $g_f(\alpha)$ at $\alpha=0$ and $\alpha=1$, and if there exists an $1\leq \alpha \leq 2$ such that $N_\alpha(f) > N_\alpha(f_0)$, then by Lemma \ref{lem1} and the continuity of $g_f(\alpha)$ the third zero must occur for some $1<\alpha<2$. In this case, the curve $g_f(\alpha)$ must be positive between $\alpha=1$ and the third zero, and 
$$\frac{\partial}{\partial\alpha}N_\alpha(f)\bigg|_{\alpha=1} > \frac{\partial}{\partial\alpha}N_\alpha(f_0)\bigg|_{\alpha=1}$$
just like in Section \ref{sec3A}. Therefore,
$$\sum_{x\in\{0,1\}^n}  T_pf(x)\log T_pf(x) > \sum_{x\in\{0,1\}^n}  T_pf_0(x)\log T_pf_0(x)$$
and the proof by contrapositive is complete.

\section{An Example $g^\mathsf{sym}_f(\alpha)$}
In order to see what these $g^\mathsf{sym}_f(\alpha)$ curves can look like, consider the following example. Let $n=3$ and
\begin{align*}
f(y) = \begin{cases} 1 & \text{, if } y_1+y_2+y_3 \geq 2 \\ 0 & \text{, otherwise} \end{cases}
\end{align*}
i.e., the majority function. For this $f$ we have
\begin{align*}
g^\mathsf{sym}_f(\alpha) & = 2\left(\left(1-p\right)^{3}+3p\left(1-p\right)^{2}\right)^{\alpha} \\
& \quad + 6\left(\left(1-p\right)^{3}+\left(1-p\right)^{2}p+2\left(1-p\right)p^{2}\right)^{\alpha} \\
& \quad + 6\left(2\left(1-p\right)^{2}p+\left(1-p\right)p^{2}+p^{3}\right)^{\alpha} \\
& \quad + 2\left(p^{3}+3\left(1-p\right)p^{2}\right)^{\alpha} \\
& \quad - 8p^{\alpha}-8\left(1-p\right)^{\alpha} \; .
\end{align*}
The theory from Lemma \ref{lem2} proves that this function will have at most four zeros. Trivial zeros occur at $\alpha=0$ and $\alpha=1$, and there must be an additional one for $\alpha>2$. In Figure \ref{fig2} we show this curve for three different values of $p$, showing three different possible behaviors.

\begin{figure}
\centerline{\includegraphics[width=8.5cm]{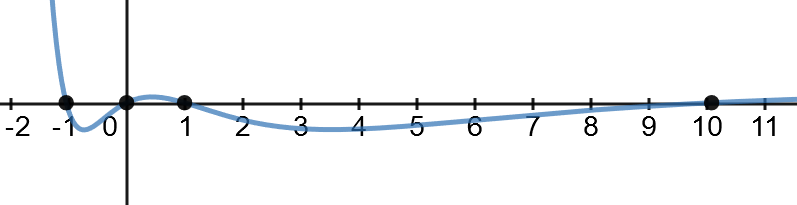}}
\vspace{.1in}
\centerline{\includegraphics[width=8.5cm]{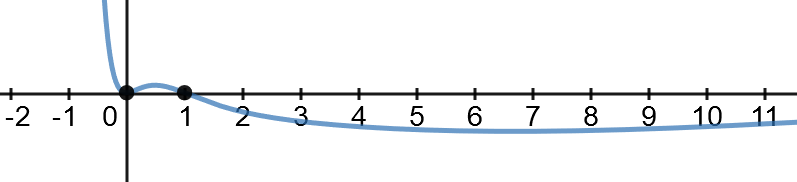}}
\vspace{.1in}
\centerline{\includegraphics[width=8.5cm]{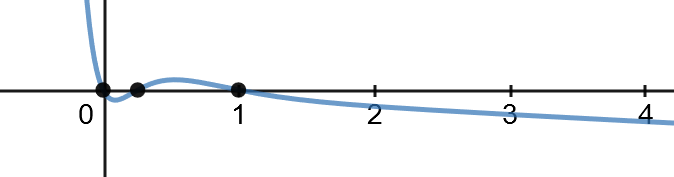}}
\caption{The $g^\mathsf{sym}_f(\alpha)$ curve as a function of $\alpha$ for the $n=3$ majority function with $p=.21$ (top), $p=.068$ (middle), and $p=.017$ (bottom). There is a zero at large $\alpha$ in the middle and bottom curves that is not visible in these plots. The fourth zero for the top curve occurs at $\alpha<0$. In the middle curve, the ``fourth zero'' is accounted for by a second order zero at $\alpha=0$. In the bottom curve, the fourth zero occurs at $0<\alpha<1$. None of the curves has a zero at $1<\alpha<2$, as this would be equivalent to $f$ being a counterexample to the Courtade-Kumar conjecture.}
\label{fig2}
\end{figure}

\section{Appendix}

In this appendix, we briefly introduce the Fourier analysis of Boolean functions in order to prove Lemma \ref{lem1}. For a full treatment we recommend \cite{boolean}. A Boolean function $f:\{0,1\}^n \to \{0,1\}$ has Fourier expansion
$$f(y) = \sum_{v\in\{0,1\}^n} \hat{f}(v)W_v(y)$$
where the set of Boolean functions $\{W_v\}_{v\in\{0,1\}^n}$ defined by
$$W_v(y) = \frac{1}{\sqrt{2^n}}(-1)^{\sum_{i=1}^n v_iy_i}$$
forms an orthonormal basis.

One way to think about the probability $\mathbb{P}(f(Y)=1|X=x)$, and the origin of the notation from \eqref{eq:noise_op}, is that it is the function of $x$ obtained via the noise operator $T_p$ applied to $f$. The function $T_pf$ can be expressed as
\begin{align*}
T_pf(x) & = \sum_{y\in\{0,1\}^n} p^{d(x,y)}(1-p)^{n-d(x,y)}f(y)
\end{align*}
where $d(\cdot,\cdot)$ is the Hamming distance, so the operator $T_p$ can also be thought of as convolution with the kernel $\varphi_p(x) = p^{d(0,x)}(1-p)^{n-d(0,x)}$. Thus using the Convolution Theorem,
\begin{align*}
T_pf(x) & = \sum_{v\in\{0,1\}^n} (1-2p)^{d(0,v)}\hat{f}(v)W_v(x)
\end{align*}
where $\lambda_v = (1-2p)^{d(0,v)}$ for $v\in\{0,1\}^n$ are the eigenvalues of the $T_p$ operator.
Using Parseval's theorem,
\begin{align} \label{eq:parseval}
N_2(f) & = \sum_{x\in\{0,1\}^n} (T_pf(x))^2 \nonumber\\
& = \sum_{v\in\{0,1\}^n} (1-2p)^{2d(0,v)}(\hat f(v))^2 \nonumber\\
& = \frac{2^n}{4} + \sum_{v\neq0}  (1-2p)^{2d(0,v)}(\hat f(v))^2
\end{align}
where the last equality \eqref{eq:parseval} follows from $f$ being balanced. Since
$$\sum_{y\in\{0,1\}} (f(y))^2 = \frac{2^n}{2} = \sum_{v\in\{0,1\}} (\hat f(v))^2$$
for all balanced $f$, it is clear that in order to maximize \eqref{eq:parseval} the remaining (non-DC-component) energy in $f$ should all be concentrated in the Fourier coefficients with $d(0,v) = 1$ , i.e., the ``weight'' one Fourier coefficients. The only balanced Boolean functions with all of their energy concentrated in the weight zero and one Fourier coefficients are the dictatorship functions, so we have
$$N_2(f) \leq N_2(f_0)$$
with equality only when $f$ is a dictatorship function. Since $f$ and $1-f$ are both dictatorship functions that maximize this quantity, we similarly have
$$N^\mathsf{sym}_2(f) \leq N^\mathsf{sym}_2(f_0) \; .$$

\bibliographystyle{ieeetran}
\bibliography{main.bib}

% Generated by IEEEtran.bst, version: 1.14 (2015/08/26)
\begin{thebibliography}{1}
\providecommand{\url}[1]{#1}
\csname url@samestyle\endcsname
\providecommand{\newblock}{\relax}
\providecommand{\bibinfo}[2]{#2}
\providecommand{\BIBentrySTDinterwordspacing}{\spaceskip=0pt\relax}
\providecommand{\BIBentryALTinterwordstretchfactor}{4}
\providecommand{\BIBentryALTinterwordspacing}{\spaceskip=\fontdimen2\font plus
\BIBentryALTinterwordstretchfactor\fontdimen3\font minus
  \fontdimen4\font\relax}
\providecommand{\BIBforeignlanguage}[2]{{%
\expandafter\ifx\csname l@#1\endcsname\relax
\typeout{** WARNING: IEEEtran.bst: No hyphenation pattern has been}%
\typeout{** loaded for the language `#1'. Using the pattern for}%
\typeout{** the default language instead.}%
\else
\language=\csname l@#1\endcsname
\fi
#2}}
\providecommand{\BIBdecl}{\relax}
\BIBdecl

\bibitem{courtade_kumar}
T.~A. Courtade and G.~R. Kumar, ``Which boolean function maximize mutual
  information on noisy inputs?'' \emph{IEEE Transactions on Information
  Theory}, vol.~60, no.~8, 2014.

\bibitem{chandar}
V.~Chandar and A.~Tchamkerten, ``Most informative quantization functions,''
  \emph{Proc. ITA Workshop, San Diego, CA, USA}, 2014.

\bibitem{or}
O.~Ordentlich, O.~Shayevitz, and O.~Weinstein, ``An improved upper bound for
  the most informative boolean function conjecture,'' \emph{Proceedings of the
  International Symposium on Information Theory (ISIT)}, 2016.

\bibitem{alex}
A.~Samorodnitsky, ``On the entropy of a noisy function,'' \emph{IEEE
  Transactions on Information Theory}, vol.~62, no.~10, p. 5446–5464, 2016.

\bibitem{muriel}
\BIBentryALTinterwordspacing
J.~Li and M.~M\'edard, ``Boolean functions: noise stability, non-interactive
  correlation distillation, and mutual information.'' [Online]. Available:
  \url{https://arxiv.org/pdf/1801.04462.pdf}
\BIBentrySTDinterwordspacing

\bibitem{hel_conj}
\BIBentryALTinterwordspacing
V.~Anantharam, A.~Bogdanov, A.~Chakrabarti, T.~Jayram, and C.~Nair, ``A
  conjecture regarding optimality of the dictator function under hellinger
  distance.'' [Online]. Available:
  \url{http://chandra.ie.cuhk.edu.hk/pub/papers/HC/hel-conj.pdf}
\BIBentrySTDinterwordspacing

\bibitem{laguerre}
\BIBentryALTinterwordspacing
E.~N. Laguerre, ``Sur la théorie des équations numériques,'' \emph{Journal
  de Math\'ematiques pures et appliqu\'ees}, 1883. [Online]. Available:
  \url{http://sepwww.stanford.edu/oldsep/stew/laguerre.pdf}
\BIBentrySTDinterwordspacing

\bibitem{mattuck}
A.~Mattuck, \emph{Introduction to Analysis}.\hskip 1em plus 0.5em minus
  0.4em\relax Prentice Hall, 1999.

\bibitem{boolean}
R.~O'Donnel, \emph{Analysis of Boolean Functions}.\hskip 1em plus 0.5em minus
  0.4em\relax Cambridge University Press, 2014.

\end{thebibliography}

\section*{Acknowledgements}
This work was supported in part by NSF award CCF-1704624 and by the Center for Science of Information (CSoI), an NSF Science and Technology Center, under grant agreement CCF-0939370.

\end{document}